\documentclass[10pt,conference]{IEEEtran}

\usepackage{bm}
\usepackage{mathrsfs}
\usepackage{textcomp}
\usepackage{epsfig}
\usepackage{indentfirst}
\usepackage{amsmath}
\usepackage{amssymb}
\usepackage{subfigure}

\usepackage{amsthm}

\newtheorem{theorem}{Theorem}
\newtheorem{lemma}[theorem]{Lemma}

\theoremstyle{definition}

\theoremstyle{remark}

\usepackage{cite}

\begin{document}
\title{Rateless Codes with Progressive Recovery for Layered Multimedia Delivery}
\author{
{Zhao Chen$^1$, Liuguo Yin$^{2,3}$, Mai Xu$^1$, and Jianhua Lu$^1$}\\
State Key Laboratory on Microwave and Digital Communications\\
Tsinghua National Laboratory for Information Science and Technology\\
$^1$Department of Electronic Engineering, Tsinghua University\\ $^2$School of Aerospace, Tsinghua University\\
$^3$EDA Lab, Research Institute of Tsinghua University in Shenzhen\\
\vspace{-0.2cm}
Email: zhao-chen10@mails.tsinghua.edu.cn; \{yinlg, xumai, lhh-dee\}@tsinghua.edu.cn}

\maketitle
\begin{abstract}
This paper proposes a novel approach, based on unequal error protection, to enhance rateless codes with progressive recovery for layered multimedia delivery. With a parallel encoding structure, the proposed Progressive Rateless codes (PRC) assign unequal redundancy to each layer in accordance with their importance. Each output symbol contains information from all layers, and thus the stream layers can be recovered progressively at the expected received ratios of output symbols. Furthermore, the dependency between layers is naturally considered.
The performance of the PRC is evaluated and compared with some related UEP approaches. Results show that our PRC approach provides better recovery performance with lower overhead both theoretically and numerically.
\end{abstract}


\section{Introduction}\label{Intro}
Rateless codes, also known as fountain codes\cite{byers2002digital}, have been proposed as one of the promising error-correcting codes recently. It has an asymptotic optimal recovery property over binary erasure channels (BEC) without any knowledge of the channel.
Compared with conventional forward error correction (FEC) codes, e.g. Reed-Solomon (RS) code, rateless codes can generate a potentially infinite stream of independent encoding symbols\footnote{Note that a symbol is the smallest encoding unit of data, which have the same size. One or more symbols will be contained in a packet.} on the fly, whereas the conventional ones must be selected a fixed code rate in advance according to the channel state. Besides, rateless codes have a lower encoding and decoding complexity.
Therefore, rateless codes can be very suitable for transmission data packets on different kinds of lossy channels, especially for fast time-varying channels or broadcast/multicast channels. In those scenarios, rateless codes can be adapted to the channel very well even with unknown or time-varying erasure probabilities, which means rateless codes can be widely applied in many applications such as Internet, mobile TV.

Multimedia contents are fragile to packet losses. And the error effects of packet losses change with their importance. Moreover, in order to achieve scalable and graceful streaming, layered multimedia codecs have been proposed in standards, such as JPEG2000\cite{Rabbani2002} and the scalable video coding (SVC)\cite{schwarz2007overview} extension of the H.264/AVC standard. In these standards, media source can be encoded into several stream layers with different importance. Lower layers, which should be decoded first, are more important and the loss of a lower layer will affect all other higher layers. So it needs to be considered to give stronger FEC protection to more important layer packets.

Luby Transform (LT) codes\cite{luby2002lt} and Raptor codes\cite{shokrollahi2006raptor}, as two state-of-the-art rateless codes, have been proved to be efficient FEC techniques to increase service robustness in multimedia delivery systems. What's more, Raptor codes have been adopted as the application layer FEC solution in current standards, such as 3GPP MBMS\cite{3GPP_MBMS} and DVB-H\cite{DVB_H}.
Nevertheless, in traditional rateless codes, all stream packets are protected by equal error protection (EEP). 
If receivers are not able to receive enough symbols, none of the source symbols can be recovered. Therefore, rateless codes with unequal error protection (UEP) are required to protect layered stream data in a more efficient way.

Recently, several rateless codes with UEP property have been proposed, which benefit in performance of the transmitted layered multimedia. In \cite{rahnavard2007rateless, cataldi2010sliding, lu2010joint, kozat2007unequal}, each layer is protected separately by assigning different redundancy according to their importance, but it is inefficient sometimes that higher layer's data may be decoded before the lower layer's, since the dependence between layers is ignored. In \cite{sejdinovic2009expanding, hellge2011layer}, with a layer-aware recovery, they sacrifice the protection for higher layers because the irregular encoding structure for high layers is inefficient, where the full recovery performance is degraded.

In this paper, a novel UEP rateless approach, named Progressive Rateless codes (PRC), is proposed. In this approach, taking the dependence into account, we reshape the stream layers and encode them parallelly without weakening the protection of higher layers. As output symbols received gradually, the stream layers can be recovered progressively one after another following their importance.

The remainder of the paper is organized as follows. In Section \ref{Back&ReWork}, we will introduce the background of rateless codes and give an overview of some related rateless UEP approaches. Section \ref{ProRaptor} generally describes our approach and conducts a theoretical probability analysis of PRC, in comparison with other approaches. Some numerical results are shown in Section \ref{SimResults}. Finally, we conclude the paper in Section \ref{Conclusion}.

\section{Background and Related Work}\label{Back&ReWork}
\subsection{Review of rateless codes}

LT codes is the first practical rateless code. Assume that we have $k$ source symbols to be transmitted.
Let $\bm{\Omega}(x)= \sum\nolimits_{i=1}^k{\Omega _i x^i}$ represent a degree distribution, $\Omega _i$ stands for the probability of degree $i$ and satisfies $\sum\nolimits_{i=1}^k{\Omega _i}=1$. The procedure of generating a encoding symbol is as follows.

\begin{enumerate}
  \item {Select an encoding degree $d$ with distribution $\bm{\Omega}(x)$.}
  \item {Choose $d$ input symbols randomly and uniformly in $k$ source symbols as neighbors of the encoding symbol.}
  \item {Perform bitwise XOR operation on the $d$ chosen symbols to generate the encoding symbol.}
\end{enumerate}

After the above procedure, an encoding symbol is transmitted to the receiver. If $d=1$, the encoding symbol is just a duplication of the unique input symbol. This procedure will be executed repeatedly and a potentially infinite encoding symbol stream can be generated until enough encoding symbols are collected at the receiver to recover all source symbols.

At the receiver, both belief propagation process (BP)\cite{luby2002lt} and maximum likelihood decoding (ML)\cite{3GPP_MBMS} can be applied to the decoding of LT codes. The procedure of BP process is as follows.

\begin{enumerate}
  \item {Initial step: search for receiving symbols with degree one and release them to recover their unique neighbor input symbols to a buffer, called the \emph{ripple}.}
  \item {Process every input symbol in the \emph{ripple} as follows until the \emph{ripple} becomes empty.
      \begin{enumerate}
        \item {Remove the input symbol from receiving symbols as a neighbor.}
        \item {Release such receiving symbols subsequently with exactly one remaining neighbor and recover their neighbors to the \emph{ripple}.}
      \end{enumerate}
        }
\end{enumerate}

BP process fails if at least one source symbol remaining unrecovered in the end. The key point of successful decoding is the perfect design of the degree distribution. Fortunately, it was proved in \cite{luby2002lt} such distribution exists and all source symbols can be recovered by any (1+$\epsilon$)$k$ encoding symbols. $\epsilon$ is the decoding overhead, it has achieved capacity-approaching behavior with very low overhead when $k\rightarrow\infty$, $\epsilon\rightarrow0$.

ML decoding, also known as full rank decoding, is executed by solving a set of linear equations in $\mathbb{F}_2^k$, since each encoding symbol is a linear combination of source symbols. It will be successful iff. the set of equations is full rank. Compared with BP process, ML decoding has higher decoding efficiency but higher decoding complexity.

Raptor codes\cite{shokrollahi2006raptor}, as an extension of LT codes, have been proposed with linear time encoding and decoding using a pre-coder of low-density parity-check (LDPC) codes. Our approach will follow Raptor codes with a modified encoding structure.

\subsection{Related work}\label{PrioriWork}

With applying rateless codes with UEP, various layered delivery techniques have been studied, which can be divided into three groups.

Rahnavard and Fekri\cite{rahnavard2007rateless}, first of all, presented a distribution-based approach. They introduced UEP at the LT encoding stage and designed a nonuniformly degree distribution such that lower layer symbols can be selected with higher probability.
With achieving unequal recovery of different layers, the altered distribution weakens the code performance and results in a larger overhead.

Another group of UEP designs are pre-coding based approaches \cite{kozat2007unequal, lu2010joint}. Without making modifications to the traditional rateless code structure, firstly layer packets are pre-coded with different code rate proportionally according to their importance, where lower layer packets are assigned to lower pre-coding rate. Then pre-coded packets are passed to a rateless encoder. Since the intermediate performance of rateless code is poor, the recovery of lower layers suffers.

The third and typical one is the redundancy-based strategy \cite{hellge2011layer, cataldi2010sliding, sejdinovic2009expanding}. As shown in Fig. \ref{Fig.SP-LA-FEC}, suppose that there are two stream layers to be delivered as \emph{Layer 2} is dependent on \emph{Layer 1}. Stream layers are encoded by two different rateless encoders and given redundant symbols proportionally with their importance. There are two types of such strategy \cite{hellge2011layer}. The separate FEC (SP-FEC) protects the two layers independently, while the layer-aware FEC (LA-FEC) extends protection following the dependency between stream layers, i.e., \emph{Layer 1} is not only protected by \emph{FEC 1} but also protected by \emph{FEC 1+2}, together with \emph{Layer 2}.
The LA-FEC improves the recovery of \emph{Layer 1} at the expense of \emph{Layer 2}, since the encoding structure of \emph{FEC 1+2} is slightly changed from original rateless codes. Next, we will have a detailed description of our approach and compare it with the redundancy-based approaches.

\begin{figure}
\centering
\subfigure[A block of a two-layer SP-FEC with $N$ output symbols.]
{
    \label{SP-FEC}
    \includegraphics[width = 7cm]{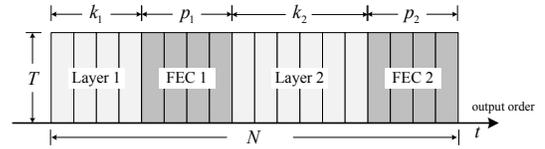}
}

\subfigure[A block of a two-layer LA-FEC with $N$ output symbols.]
{
    \label{LA-FEC}
    \includegraphics[width = 7cm]{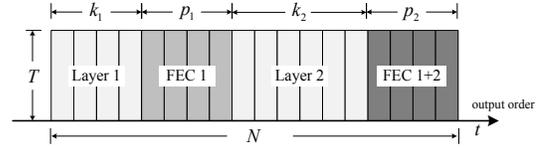}
}
\caption{Structure of two types of separate unequal protection strategy, whose $N$ output symbols are transmitted from left to right sequentially.}
\label{Fig.SP-LA-FEC} 
\vspace{-1.0em}
\end{figure}

\begin{figure}
\centering
\includegraphics[width=7.0cm]{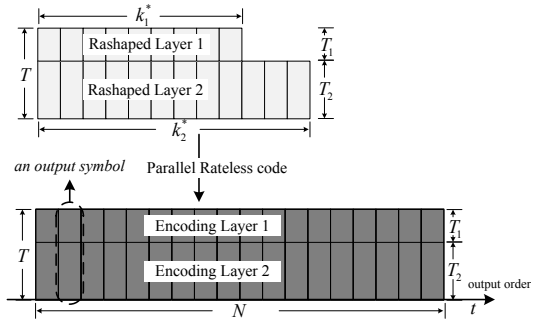}
\vspace{-1.0em}
\caption{Encoding procedure of a block of a two-layer PRC, whose $N$ output symbols are transmitted from left to right sequentially.}
\label{Fig.ProRateless}
\vspace{-1.0em}
\end{figure}

\section{Progressive Rateless Codes}\label{ProRaptor}
In this section, we propose the Progressive Rateless code (PRC) to enhance traditional rateless codes with UEP capability.
In our approach, to guarantee the optimized recovery performance, we alter the encoding structure with maintaining the parameters of rateless code, e.g. degree distribution. Meanwhile, with the efficient recovery of rateless codes, the dependency between layers has been satisfied to come to a progressive recovery of the layered multimedia.


\subsection{Design and Implementation}\label{DeProRaptor}
We consider a layered multimedia data stream to be transmitted over an erasure channel. Assume that an $m$-layer video stream is partitioned into several source blocks with the size of $K$ symbols, where the importance of symbols decreasing from \emph{Layer $1$} to \emph{Layer $m$}. Let $k_i$ be the number of source symbols of \emph{Layer $i$}, so that $K = \sum\nolimits_{i = 1}^m {k_i}$. Let $T$ be the symbol length in bytes, thus each layer has $k_i \cdot T$ bytes and the total length of the block will be $K \cdot T = \sum\nolimits_{i = 1}^m {k_i \cdot T}$ bytes. Note that $\frac{k_i}{K}$ is a constant for \emph{Layer i} in a certain layered stream as the source block size $K$ changes.

Given total broadcasting bandwidth, the overall coding rate $\gamma = \frac{K}{N}$ is fixed for all possible $K$, where $N$ is the output block size, i.e. number of output symbols for each source block, to protect the layered data stream from packet losses. Thus the total length of redundant symbols is $\left( N - K \right) \cdot T$ bytes. Based on these conditions, our PRC approach will generate encoding symbols in a parallel way.

Before rateless encoding, all $m$ layers are reshaped with symbol lengths of $\{T_1,T_2,...,T_m\}$ bytes respectively, ensuring $T = \sum\nolimits_{i = 1}^m {T_i}$. Then the number of reshaped symbols in \emph{Layer i} becomes $k_i^{\ast} = \frac{k_i \cdot T }{T_i}$. Each reshaped layer is passed through a rateless encoder to generate $N$ reshaped encoding symbols, where an output symbol is formed by packing $m$ reshaped encoding symbols, one from each encoding layer. So there will be $N$ output symbols with each symbol packing encoding data from all layers. As shown in Fig. \ref{Fig.ProRateless}, a two-layer PRC layered delivery is illustrated, where an output symbol is generated by combining two reshaped encoding symbols.

At the decoder, assume that $R$ output symbols are received, of course $R \leq N$ due to packet losses. The received symbols are firstly unpacked to separate reshaped symbols of each layer, which are then passed to $m$ different rateless decoders, respectively. Lastly the message block are recovered layer by layer at the decoders.

\subsection{Recovery Performance Analysis}\label{Analysis}
In this part, we will make an combinational analysis of recovery probability of PRC, in comparison with SP-FEC and LA-FEC in Fig. \ref{Fig.SP-LA-FEC}.
To make a fair comparison, for \emph{Layer i}, we have redundant data of the equal length in all approaches, i.e., $p_i \cdot T = (N - k_i^{\ast}) \cdot T_i$, where $p_i$ is the number of redundant symbols. And then we have the number of output symbols $n_i = k_i + p_i$, the coding rate $r_i = \frac{k_i}{n_i}$. Thus we obtain
\begin{equation}\label{K_star}
    k_i^{\ast} = \frac{k_i \cdot T}{T_i} = \frac{k_i \cdot T}{n_i \cdot T / N} = N \cdot \frac{k_i}{n_i}= N \cdot r_i.
\end{equation}
which shows the coding rate of reshaped \emph{Layer i} is also $r_i$. In SP-FEC, let $\eta_i = \frac{n_i}{N}$ be the output ratio of \emph{Layer i}, which will be a constant once $r_i$ is determined.

Let $\textrm{Pr}_{i}(R)$ be the recovery probability of \emph{Layer i} in PRC. Without loss of generality, we consider the ideal recovery of rateless codes, i.e., $k_i^{\ast}$ source symbols can be recovered as soon as $R \geq k_i^{\ast}$ encoding symbols are received. Then we have
\begin{equation}\label{Equation:ProRateless}
    \textrm{Pr}_{i}(R) =
    \begin{cases}
        1, & \hbox{$R \geq k_i^{\ast}$;} \\
        0, & \hbox{$R < k_i^{\ast}$.}
    \end{cases}
\end{equation}
which indicates that for \emph{Layer i}, it can be recovered from at most  $N - k_i^{\ast}$ packet losses with probability 1.
Therefore, to recover layered data stream progressively from \emph{Layer 1} to \emph{Layer m}, we can make $k_1^{\ast} \leq k_2^{\ast} \leq ... \leq k_m^{\ast}$. From Eq. (\ref{K_star}), we know that it can be also represented by $r_1 \leq r_2 \leq ... \leq r_m$. So we can protect \emph{Layer i} by assigning suitable $r_i$, expecting to recover it after receiving $r_i$ ratio of output symbols.

Following the results in \cite{hellge2011layer}, we only investigate the SP-FEC without considering dependency for simplicity. The analysis of LA-FEC will be similar to that and when $N$ grows large, there will be little difference between the two approaches.

In SP-FEC, let $\mathrm{Pr}_{i}^{'}(R)$ be the recovery probability of \emph{Layer i}. With the ideal recovery assumption of rateless codes, it can be recovered by at least $k_i$ out of $n_i$ output symbols from \emph{Layer i}.

\begin{itemize}
  \item {For $R < k_i$, $\mathrm{Pr}_{i}^{'}(R) = 0$.}
  \item {For $k_i \leq R < N - (n_i - k_i)$, 
  \begin{equation}\label{Equation:SP-FEC}
    \mathrm{Pr}_{i}^{'}(R) = {\sum\limits_{x={k_i}}^{{\min\{R,n_i\}}}
    \frac{\binom{n_i}{x} \binom{N-n_i}{R-x}} {\binom{N}{R}}}
    = {\sum\limits_{x={k_i}}^{{\min\{R,n_i\}}}
    \frac{\binom{R}{x} \binom{N-R}{n_i-x}} {\binom{N}{n_i}}}
  \end{equation}}
  \item {For $R \geq N - (n_i - k_i)$, $\mathrm{Pr}_{i}^{'}(R) = 1$.}
\end{itemize}

It's clear that $\mathrm{Pr}_{i}^{'}(R)$ is the tail probability of a hypergeometric distribution with parameters of $X \sim \mathcal{H}(n_i,R,N)$, i.e.,
\begin{equation}\label{Eq:Pr_i_Definition}
    \mathrm{Pr}_{i}^{'}(R) = P(X \geq k_i | n_i,R,N)
\end{equation}

\begin{figure}
\centering
\includegraphics[height = 6cm]{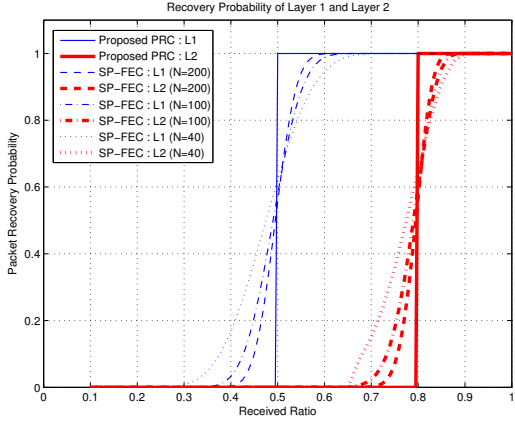}
\caption{Recovery probability of proposed PRC and SP-FEC with receiving different ratio of output symbols. Note that L1 and L2 denote \emph{Layer 1} and \emph{Layer 2}, respectively. The coding rate for the two layers are  $r_1 = 1/2$ and $r_2 = 4/5$.}
\label{Fig.CompareTheory}
\vspace{-1.0em}
\end{figure}

Let $r = R/N$ be the received ratio of all output symbols. In Fig. \ref{Fig.CompareTheory}, several curves of recovery probability are shown, normalized by the received ratio. Note that in practise typically multimedia codecs work well under a packet loss rate (PLR) of no more than $10^{-4}$, so the recovery performance can be measured by the \emph{Successful Received Ratio (SRR}), where the recovery probability goes above $1 - 10^{-4}$. From Fig. \ref{Fig.CompareTheory}, we can see the SRR of PRC is much smaller than SP-FEC for both \emph{Layer 1} and \emph{Layer 2}.

To find the relationship between the two approaches, we will show some properties of $\mathrm{Pr}_{i}^{'}(R)$ as the following lemmas.

\begin{lemma}\label{Le:Increase}
$\mathrm{Pr}_{i}^{'}(R)$ is a non-decreasing function of the number of received symbols $R$.
\end{lemma}

\begin{proof}
From Eq. (\ref{Equation:SP-FEC}), if $R < k_i$ or $R \geq N - (n_i - k_i)$, $\mathrm{Pr}_{i}^{'}(R)$ is a constant.
Otherwise we can obtain
\begin{align*}
\mathrm{Pr}_{i}^{'}(R+1) & = \frac{ \hspace{-0.5em}\sum\limits_{x={k_i}}^{{\min\{R,n_i\}}} \hspace{-0.5em}  {\binom{R}{x} \binom{N-R-1}{n_i-x}}
+ \hspace{-0.5em} \sum\limits_{x={k_i-1}}^{{\min\{R,n_i-1\}}} \hspace{-0.5em}  {\binom{R}{x} \binom{N-R-1}{n_i-x-1}}}{\binom{N}{n_i}} \\
 & =  \mathrm{Pr}_{i}^{'}(R) + \frac{\binom{R}{k_i-1}\binom{N-R-1}{n_i-k_i}}{\binom{N}{n_i}}
\end{align*}
which shows $\mathrm{Pr}_{i}^{'}(R)$ is increasing strictly when $k_i \leq R < N - (n_i - k_i)$.
\end{proof}
Lemma \ref{Le:Increase} shows the recovery probability of any \emph{Layer i} increases with the number of the received symbols, which is in accordance with our intuition.
Furthermore, given a fixed received ratio of output symbols, we have Lemma \ref{Le:IncreaseWithN} when the output block size $N$ increases with a fixed overall coding rate.

\begin{lemma}\label{Le:IncreaseWithN}
If $r$ is any constant and $r > r_i$, $\mathrm{Pr}_{i}^{'}(N \cdot r)$ is an increasing function of the output block size $N$.
\end{lemma}

\begin{proof}
With the property of hypergeometric distribution, the mean $\mu$ and the variance $\sigma^2$ of $X$ are given by
\begin{equation}
\left. \mu = R \cdot \frac{n_i}{N}, \sigma^2 = N \cdot r(1-r)\frac{n_i}{N}(1-\frac{n_i}{N}) \right.
\end{equation}

If $N$ is large enough, $X$ can be approximated by a normal distribution, if the following conditions can be satisfied,
\begin{enumerate}
  \item $\frac{R}{N} \rightarrow r$ is a constant.
  \item $\frac{n_i}{N} \rightarrow \eta_i$ is a constant.
\end{enumerate}
Then $X$ approaches to $\mathcal{N}(\mu,\sigma^2)$, where
\begin{equation}\label{Eq:Approximate}
P(X=x | R,n_i,N) \approx \frac{1}{\sqrt{2\pi}\sigma}e^{-\frac{(x-\mu)^2}{2\sigma^2}}
\end{equation}
Since $k_i = n_ir_i = N \eta_i r_i$, we can derive with Eq. (\ref{Eq:Pr_i_Definition}),
\begin{equation}\label{Eq:ApproximatePr_i}
    \mathrm{Pr}_{i}^{'}(Nr) \approx 1 - \Phi\left(\frac{k_i - \mu}{\sigma}\right) = \Phi\left(\frac{(r-r_i)\sqrt{N}}{\sqrt{r(1-r)(1-\eta_i)}} \right)
\end{equation}
where $\Phi(x)$ is the cumulative distribution function of the standard normal distribution. With the monotonicity property of $\Phi(x)$, it's straightforward to show that $\mathrm{Pr}_{i}^{'}(Nr)$ increases with $N$ if $r > r_i$.
\end{proof}

Lemma \ref{Le:IncreaseWithN} shows that the recovery probability of each \emph{Layer i} will increase with $N$ when a certain ratio of symbols are received, as long as the ratio is more than $r_i$. If the PLR is no more than $1-r_i$, we can improve the performance by assigning a lager block size $K$.

\begin{lemma}\label{Le:Infinity}
If $r$ is any constant and $r > r_i$, when output block size $N \rightarrow \infty$, \begin{equation}\label{Eq:Infnity}
    \lim_{N\rightarrow \infty}{\mathrm{Pr}_{i}^{'}(N \cdot r)}=\mathrm{Pr}_{i}(N \cdot r)
\end{equation}
\end{lemma}

\begin{proof}
Firstly we will show that $\mathrm{Pr}_{i}^{'}(N \cdot r) \rightarrow 1$.
This is quite easy to be shown from Eq. (\ref{Eq:ApproximatePr_i}), since $\Phi(x) \rightarrow 1$ when $x$ goes to infinity.

Recall Eq. (\ref{K_star}), if $r > r_i$, $R > Nr_i = k_i^{\ast}$. So we have $\mathrm{Pr}_{i}(N \cdot r) = 1$. Thus we we conclude the assertion.
\end{proof}

Lemma \ref{Le:Infinity} shows the asymptotic recovering probability of SP-FEC is equal to PRC, which means PRC seems to be optimal for the SP-FEC. We also notice that the SRR of SP-FEC approaches to that of PRC when $N$ grows. In other words, for moderate output block size of $N$, PRC will theoretically outperform SP-FEC with lower overhead. Next, we will show some numerical results to confirm those conclusions.

\section{Simulation Results}\label{SimResults}

In order to verify the asymptotic analysis developed in Section \ref{Analysis}, we performed numerical simulations of PRC compared with SP-FEC and LA-FEC \cite{hellge2011layer}.
In our simulations, we apply the Raptor codes specified in \cite{3GPP_MBMS} as the rateless codec in all approaches, which has achieved an efficient recovery with a very low overhead.

Assume that we have a two-layer multimedia data stream, e.g. a MPEG TS stream with a base layer and an enhanced layer, to be transmitted. The symbol length is $T = 188$ bytes. The layered steam data can be partitioned into message blocks of different $K$, where $\frac{k_1}{K} = \frac{5}{13}$ and $\frac{k_2}{K} = \frac{8}{13}$ are constant. To have a progressive recovery of the layered stream, we assign $r_1 = 0.5$ and $r_2 = 0.8$ for the two layers.

We have performed two cases on $(N = 1000,K = 650)$ and $(N = 500,K = 325)$ to evaluate the PLR of the two layers after recovering, shown in Fig. \ref{Fig.SimL1} and Fig. \ref{Fig.SimL2} respectively. It can be seen in the figures that both \emph{Layer 1} and \emph{Layer 2} of the PRC are recovered around $50\%$ and $80\%$ as expected with a very low overhead $2\%(N=1000)$ to $4\%(N=500)$. Thus we have achieved a progressive recovery of the layered stream.

Fig. \ref{Fig.SimL1} shows the PLR of \emph{Layer 1} of all the three approaches with the received ratio near $50\%$. It's clear that the PRC outperforms both SP-FEC and LA-FEC with reducing more than $5\%(N=1000)$ to $7\%(N=500)$ received symbols below the PLR of $10^{-4}$. When the output block size $N$ increases, the PLR at the same received ratio decreases. Moreover, the gap between PRC and the other two approaches becomes closer, which meet our conclusions in Lemma \ref{Le:IncreaseWithN} and Lemma \ref{Le:Infinity}. Note that the curves of SP-FEC and the LA-FEC are too close to tell, implying that the recovery performance of the two approaches are nearly the same. The reason is the efficient recovery of Raptor codes, i.e., \emph{Layer 1} has been almost recovered before \emph{Layer 2} is just to be recovered.

Fig. \ref{Fig.SimL2} demonstrates the PLR of \emph{Layer 2} with about $80\%$ output symbols received. It also indicates the advantage of PRC as in Fig. \ref{Fig.SimL1}. However, we can see in Table \ref{T:RecoverTime}, using the PRC, there will be a little longer delay for \emph{Layer 1}, while the buffering time for \emph{Layer 2} is saved. Consequently, our approach can provide each layer with a progressive recovery according to its importance. Moreover, it reduces the received overhead efficiently compared with the other UEP approaches.


\begin{figure}
\centering
\includegraphics[height = 6cm]{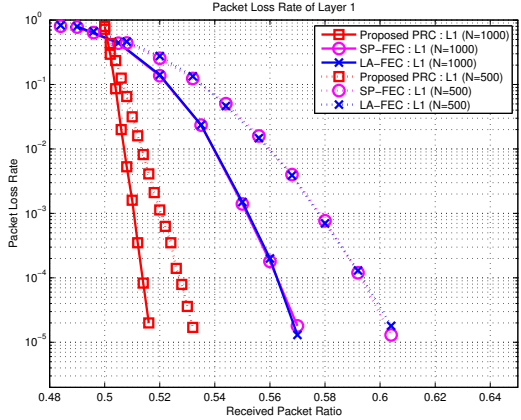}
\caption{Packet loss rate of \emph{Layer 1} of the proposed PRC versus other UEP approaches after receiving a certain ratio of output symbols. Both $N = 1000$ and $N = 500$ are illustrated.}
\label{Fig.SimL1}
\vspace{-1.0em}
\end{figure}
\begin{figure}
\centering
\includegraphics[height = 6cm]{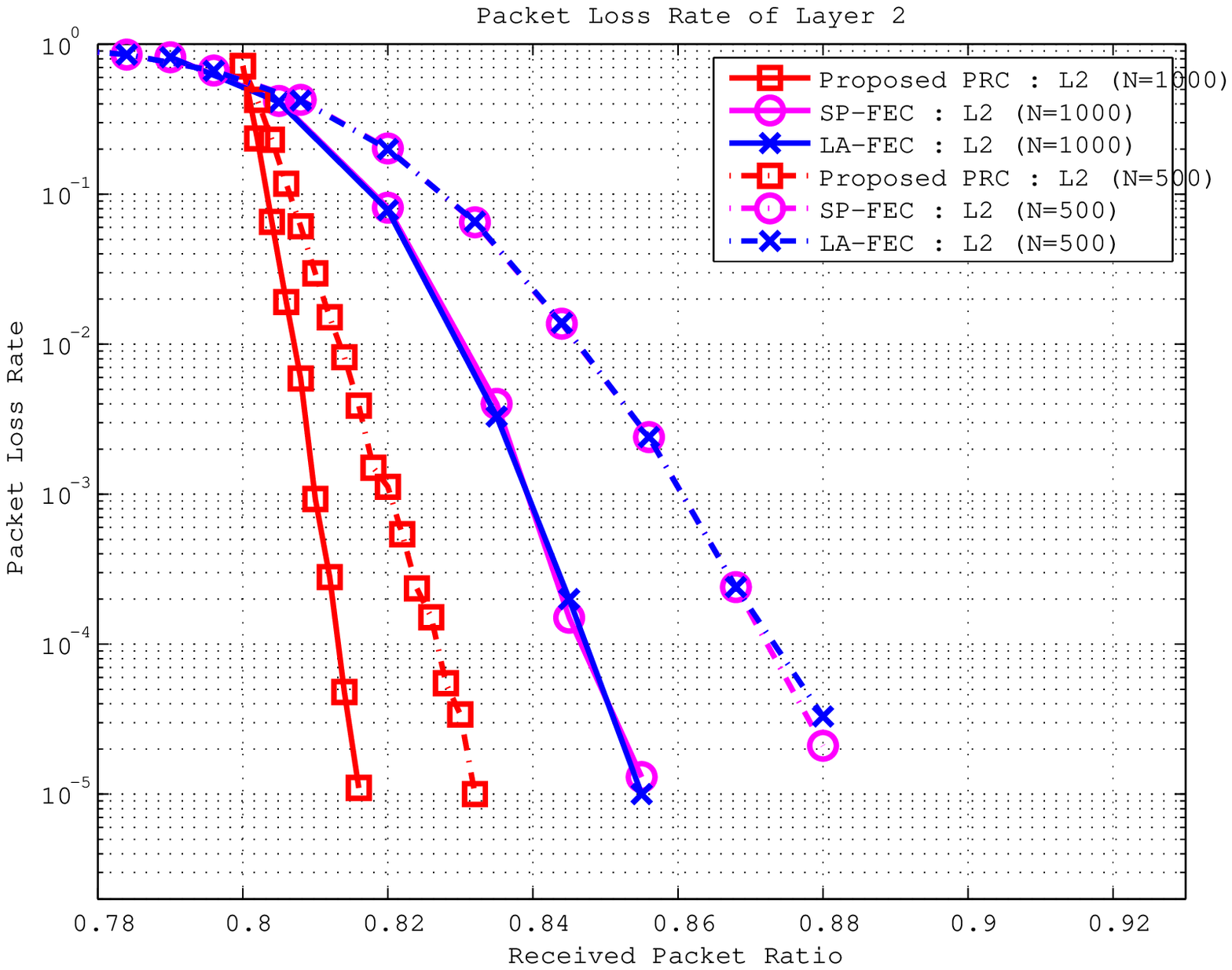}
\caption{Packet loss rate of \emph{Layer 2} of the proposed PRC versus other UEP approaches after receiving a certain ratio of output symbols. Both $N = 1000$ and $N = 500$ are illustrated.}
\label{Fig.SimL2}
\vspace{-1.0em}
\end{figure}


\begin{table}
\caption{Average buffering time to recover a message block, normalized by the time to send a symbol.}
\label{T:RecoverTime}
\centering
\begin{tabular}{ |c |c | c| c| c |}
  \hline
      & \multicolumn{2}{c|}{$N = 1000$}& \multicolumn{2}{c|}{$N = 500$} \\\cline{2-5}
      & \emph{Layer 1} & \emph{Layer 2} & \emph{Layer 1} & \emph{Layer 2} \\
  \hline
  PRC & 501.60 & 801.59 & 251.63 & 401.62 \\
  \hline
  SP-FEC & 251.62 & 901.62 & 127.01 & 451.49 \\
  \hline
  LA-FEC & 252.52 & 902.21 & 127.80 & 451.92 \\
  \hline
\end{tabular}
\vspace{-1.0em}
\end{table}

\section{Conclusion}\label{Conclusion}
In this paper we have presented a novel rateless codes with progressive recovery, on the basis of unequal error protection for layered multimedia delivery. Considering the different importance of layers of multimedia data, we protect each layer with unequal redundancy. Each output symbol in the proposed approach is packed with encoding data from all layers, whereas conventional approaches pack output symbols separately. Then the proposed approach shows promise of recovering the layered stream progressively according to the dependency at designated received ratio. In comparison with different related UEP approaches, theoretical and numerical results suggest the superiority of the proposed approach with reduced overhead of all stream layers, indicating that our approach can be widely employed to multimedia delivery implementations, especially for layered multicast or broadcast.

\section{Acknowledgment}
This research is supported in part by the NSFC under Grants No.61101072 and No.61021001, China Postdoctal Science Foundation under Grant No.2011M500327, the National Key Technology R\&D Program of China (2008BAH25B03).

\bibliographystyle{IEEEtran}
\bibliography{IEEEfull,CZBib}

\end{document}